\documentclass[12pt]{article}

\usepackage{setspace}
\usepackage{fullpage}
\usepackage[T1]{fontenc}
\usepackage{ae,aecompl}
\usepackage[dvips]{graphicx}

\usepackage[round,authoryear]{natbib}
\usepackage{color}
\usepackage{array}
\usepackage{rotating}
\usepackage{colortbl}
\usepackage{tikz}
\usetikzlibrary{patterns}

\usepackage{amsthm}
\usepackage{amsmath}
\usepackage[utf8]{inputenc}
\usepackage{changepage}
\usepackage{enumitem}
\usepackage[toc,page]{appendix}
\usepackage[margin=1in, footnotesep=1cm]{geometry}
\usepackage{mathabx}

\definecolor{webgreen}{rgb}{0,0.4,0}
\definecolor{webbrown}{rgb}{0.6,0,0}
\definecolor{purple}{rgb}{0.5,0,0.25}
\definecolor{darkblue}{rgb}{0,0,0.7}
\definecolor{darkred}{rgb}{0.7,0,0}
\definecolor{darkgreen}{rgb}{0,0.7,0}
\usepackage[pdfborder=false]{hyperref}
\hypersetup{colorlinks,citecolor=darkred,filecolor=black,linkcolor=darkblue,urlcolor=webgreen,pdfpagemode=none,
pdfstartview=FitH}

\usepackage{sectsty}
\sectionfont{\centering\normalfont\scshape}
\subsectionfont{\centering\normalfont}

\begin{document}
\begin{spacing}{1.3}

\newtheorem{definition}{Definition}[]
\newtheorem{question}{Question}[]
\newtheorem{observation}{Observation}[]
\newtheorem{conjecture}{Conjecture}[]
\newtheorem{claim}{Claim}[]
\newtheorem{lemma}{Lemma}[]
\newtheorem{proposition}{Proposition}[]
\newtheorem{theorem}{Theorem}[]
\newcounter{example}[section]
\newenvironment{example}[1][]{\refstepcounter{example}\par\medskip
   \noindent \textbf{Example~\theexample. #1} \rmfamily}{\medskip}

\title{Single-peaked domains with designer uncertainty \footnote{Guidance by Arunava Sen and Debasis Mishra is gratefully acknowledged. I am also thankful to Stephen Morris and Alex Wolitzky for insightful comments.}}

\author{Aroon Narayanan \footnote{Department of Economics, MIT}}

\date{}

\maketitle

\begin{abstract}
This paper studies single-peaked domains where the designer is uncertain about the underlying alignment according to which the domain is single-peaked. The underlying alignment is common knowledge amongst agents, but preferences are private knowledge. Thus, the state of the world has both a public and private element, with the designer uninformed of both. I first posit a relevant solution concept called implementation in mixed information equilibria, which requires Nash implementation in the public information and dominant strategy implementation in the private information given the public information. I then identify necessary and sufficient conditions for social rules to be implementable. The characterization is used to identify unanimous and anonymous implementable social rules for various belief structures of the designer, which basically boils down to picking the right rules from the large class of median rules identified by Moulin (1980), and hence this result can be seen as identifying which median rules are robust to designer uncertainty.\\


\end{abstract}

\newpage

\section{Introduction}

Consider the following situation: a group of people, say voters, need to select an alternative, say a political party, amongst many that are available to them. An independent authority is established to determine what the selection rule should be i.e. what alternative to select based on the preferences of each person in the group. In many such situations, it is reasonable to assume that the agents, probably after much discussion amongst themselves, have a common concern, such as national defence. They would each have some ideal amount of national defense spend in mind, which would correspond to their ``peak", and they would find any amount farther away from this peak less attractive than amounts that are closer to it. This captures their preference over the alternatives. In evaluating the parties based on their position on how much to spend on national defense, the voters thus find themselves having ``single-peaked preferences" over them. This basically means that if a party proposes a spend closer to a voter's peak, the voter will like it more.

First introduced by \cite{black}, single-peaked domains were studied extensively by \cite{Moulin1980}. Since then, the literature has studied them in the context of some properties of rules that seem quite natural: first, everybody should find it optimal to report their true preference regardless of what the others say they prefer, second, if everyone agrees on an alternative, then that should be chosen, and third, every agent should be ex-ante equal. These are the properties of strategy-proofness, unanimity, and anonymity, respectively. A nice class of rules - the median rules - satisfy these conditions, and in fact only they satisfy these conditions. Essentially such rules pick the median of the $n$ reported peaks and $n-1$ fixed ``phantom peaks", using an exogenous underlying alignment according to which the preferences of the agents are known to be single-peaked. Each median rule can be uniquely identified by these $n-1$ ``phantom peaks". \cite{Moulin1980} also showed that an even more general characterization of  rules - which \cite{sprumont1995} terms min-max rules - can be obtained if one drops anonymity.

However, in this analysis it was important to assume that the underlying alignment according to which everyone evaluates the alternatives is known to the authority. In many situations, this may not be reasonable since the authority may not be privy to such information. For example, when the legislature sets up a committee to decide funding for a program, it is quite likely that the committee members will arrive at a common concern after discussions, and then vote for the optimal funding based on that concern. But it is unlikely that the legislature would know the specific concern beforehand. It would then be prudent to set up a rule that accounts for this uncertainty. It may even be the case that priorities that produce such alignments amongst voters changes drastically over time. Witness the remarkable realignment in the United States Democratic Party's position on race, from being anti-abolitionist in the 19th century to passing the Voting Rights Act in the 20th century. It is reasonable to expect that long-lasting rules be designed keeping such shifting alignments in mind. Applications also extend to similar settings in auctions, regulation and a variety of other domains.

The objective of this paper is to identify social rules that will be robust to the designer's uncertainty about the underlying ordering. In particular, we would like to relate this to Moulin's class of median rules. Given $m$ alternatives and $n$ players, his median class contains ${m+n-2}\choose{n-1}$ rules for each possible underlying alignment. Clearly a designer then faces a problem of plenty, and we would like our results to help her in identifying which of these will have this additional property i.e. the property of being robust to uncertainty about the underlying ordering.

In order to achieve this objective, it is imperative to first formalize what this robustness notion should be. Note that the information of each agent can be split into two - a public state, which is the underlying alignment, and a private state, which is his own preference. The state of the world then is composed of an alignment that is in the support of the designer's belief and a profile of preferences that is single-peaked according to this alignment, and any SCF picks an alternative for each state. Given the mixed information nature of our setting, we posit a suitable solution concept for implementation of SCFs. With respect to the underlying alignment, which is common knowledge amongst the agents, we require Nash implementation. Then, given that the underlying alignment is being truthfully revealed, we require implementation in dominant strategies with respect to the private information i.e. their preferences.

Of course, applying our solution concept directly to rules is cumbersome, and hence first we must obtain a characterization of this concept in terms of easy-to-check conditions on social rules. By the first requirement, implementable SCFs will be strategy proof for states that share the same alignment. By the second, they will have to satisfy a condition that we called shared-monotonicity, which requires that across two states that share a preference profile, the SCF must choose the same outcome at that profile. This is due to the lower contour sets of agents matching exactly for agents with such preferences across the two states, which leads to Nash equilibria in one state being Nash equilibria in the other. These conditions, with no veto power, also turn out to be sufficient.

This characterization enables us to carry out our original objective - identify robust unanimous and anonymous SCFs. In fact, this objective had a greater role in inspiring the kind of solution concept that we use than may be apparent from the deductive narrative we have laid out so far. The axioms characterizing the solution are such that adding unanimity and anonymity implies that implementable SCFs, fixing an alignment from the set of all alignments that the designer believes plausible, must be Moulin's median rules. We can then check which median rules satisfy the additional conditions imposed by implementability for some interesting belief structures of the designer. This then enables us to cleanly answer some of the question of robustness within Moulin's class that inspired us. In fact, this will form the bulk of our study. The most important of these is the case when the designer believes that the preferences may have been generated by any possible alignment - in this case, only the true median rule, i.e. the rule that picks the median of the reported peaks of the agents, survives the demanding notion of implementability. We also derive results for other interesting forms of designer uncertainty.

We are not aware of any paper that studies settings in which the state of the world has both public and private elements. In that sense, our setting is novel. There are papers that look at domains which feature multiple single-peaked domains devoid of any asymmetry in information between the designer and the agents about the alignments, such \cite{REFFGEN2015349}. Another related strand of literature studies single-peaked preferences on multi-dimensional domains, as in \cite{barb}, \cite{border}, and \cite{chich}.

\section{Model}

The environment is the tuple $(N, X, \mathcal{P}, \mu)$ where $N := \{1, . . . , n\}$ is a set of $n$ agents, $X$ is a finite set of alternatives, $\mathcal{P}$ is the set of all strict alignments over $X$, and $\mu$ is a non-degenerate probability distribution over $\mathcal{P}$. We assume $n \geq 3$.

The agents in this setting have single-peaked preferences according to some alignment in the support of $\mu$ \footnote{Note that any alignment and its exact reverse, for example $a \succ b \succ c$ and $c \succ b \succ a$, produce the same single-peaked domain, and hence are equivalent for our purposes. All claims pertain to this equivalence class.}, and this alignment is common knowledge amongst them. $\mu$ represents the designer's belief over which possible alignment could be underlying the agents' preferences. The preferences of agents are private knowledge, i.e. unknown to both other agents and the designer.  

Throughout this study, we will attempt to derive results that do not depend heavily on the specific prior the designer has. This will be achieved by working with only the support of the prior, so that our results will be robust to small changes in the probability assigned to the alignments in the support of the prior.

We assume wlog that the message space of each agent $i$ in any mechanism must be of the form $M_i = L_i \times \mathcal{K}$ where $L_i$ is allowed to be any arbitrary space, and $\mathcal{K} \subseteq \mathcal{P} \times \mathcal{P}$ contains elements of the form $(\succ, P_i)$ where $P_i$ is single-peaked according to $\succ$.

\begin{definition}
A \textbf{mechanism} is the tuple $(M,g)$ where $M = (M_i)_{i \in N}$ is the message space and $g:M \to X$ assigns an outcome to each message profile. 
\end{definition}

Let $\Theta$ denote the set of all states. Then, for any state $\theta \in \Theta$, we can represent it as $\theta = (\succ^{\theta}, (P^{\theta}_{i})_{i \in N})$, where $\succ^{\theta}$ is an alignment in the support of $\mu$ and represents the public state while $P^{\theta}_{i}$ represents the private state of agent $i$. The tuple $(P^{\theta}_{i})_{i \in N}$ identifies a profile of preferences that is single-peaked according to the public state.

\begin{definition}
A \textbf{social choice function} $f:\Theta \to X$ assigns an alternative to each state $\theta$.
\end{definition}

Finally, we formally define the solution concept that we set forth in the introduction. The idea is simple and intuitive - we require Nash implementation in the public component and dominant strategy implementation in the private component of the state given truthful reporting of the public component. Essentially, Nash implementation in the public component ensures nobody has an incentive to unilaterally misreport the common ordering. Then, given the common ordering is being reported truthfully, dominant strategy implementation ensures that the private preferences are also reported truthfully.

\begin{definition}
A mechanism $(M,g)$ implements a social choice function $f$ in \textbf{mixed information equilibria} if for every state $\theta = (\succ^{\theta}, (P^{\theta}_{i})_{i \in N})$ there exists a message $m^* = (l^{*}_i, \succ^{\theta}, (P^{\theta}_{i})_{i \in N})$ such that $g(m^*) = f(\theta)$ and we have:
\begin{enumerate}

    \item Nash implementation in public information:
    
    \begin{itemize}
        \item for all $i$, $$g(m^{*}_{i}, (l_{-i}, \succ^{\theta}, P^{\theta}_{-i})) \ R^{\theta}_i \ g((l_{i}, \succ^{'}, P^{'}_{i}), (l_{-i}, \succ^{\theta}, P^{'}_{-i}))$$ for all $(l_i)_{i \in N}$, for all $(\succ^{'},P^{'}_{i}) \in \mathcal{K}$ 
        \item if $\Bar{m}$ is a Nash equilibrium and $\Bar{m}_i = (\Bar{l}_i, \Bar{\succ}_i, P^{\theta}_{i})$ for some $(\Bar{l}_i, \Bar{\succ}_i)$, then $g(\Bar{m}) = g(m^*)$
    \end{itemize}

    \item Dominant strategy implementation in private information given truthful public state reporting: for all $i$, $$g(m^{*}_{i}, (l_{-i}, \succ^{\theta}, P_{-i}^{'})) \ R^{\theta}_i \ g((l_{i}, \succ^{\theta}, P_{i}^{'}), (l_{-i}, \succ^{\theta}, P_{-i}^{'}))$$ for all $(l_i)_{i \in N}$, for all $P_{i}^{'}$, for all $P_{-i}^{'}$, with strict preference for some $(l_{-i}, \succ^{\theta}, P_{-i}^{'})$

\end{enumerate}

A social choice function $f$ is said to be \textbf{implementable} in mixed information equilibria if there exists a mechanism that implements it.

\end{definition}

To further expand on this definition, the first condition of Nash implementation requires that agents not be able to unilaterally deviate along both private and public dimensions while the second ensures that the equilibrium message is the unique Nash equilibrium with respect to the public state. This ensures that the public state is conveyed accurately to the designer. Then, the second condition says that given this public state, we must have dominant strategy implementation in the private state.

\section{Implementable choice functions}

The primary motive of this study is to understand the nature of choice functions which can be implemented by the designer in this context. Clearly the definition above of implementability is onerous to check in practice. The obvious first step then would be to arrive at a characterization in terms of some properties that can be directly verified for choice functions, and this is indeed the step we take. Then, we can apply these properties directly to the choice functions to identify which of them satisfy these properties.

Our first result will be to identify necessary and sufficient conditions for SCFs to be implementable in mixed information equilibria. In order to state the result, a few notions will be required.

Denote by $f_{\succ}$ the SCF $f$ restricted to $\{\theta = (\succ, P)\}$ where $P$ is allowed to be any profile of preferences that is single-peaked according to $\succ$.

\begin{definition}
The SCF $f_{\succ}$ is \textbf{strategy-proof} if for every preference $P_i$ that is single-peaked according to $\succ$, either $f_{\succ}(P_i,P_{-i})$ $P_i$ $f_{\succ}(P_i^{'}, P_{-i})$ or $f_{\succ}(P_i,P_{-i}) = f_{\succ}(P_i^{'}, P_{-i})$ for all deviating preference reports $P_i^{'}$ and for all preferences of the other agents $(P_{-i})$.
\end{definition}

\begin{definition}
The SCF $f$ is \textbf{shared-monotonic} if for all pairs of alignments $\succ, \succ^{'}$ such that there exists a profile of preferences $P$ that is single-peaked according to both $\succ$ and $\succ^{'}$, $f(\succ, P) = f(\succ^{'}, P)$.
\end{definition}

\begin{definition}
The SCF $f$ satisfies \textbf{no veto power (NVP)} if for all $\theta = (\succ, P)$ such that $P_i(1) = a$ for all but one agent, then $f(\theta) = a$.
\end{definition}

Before we present the theorem, let us try to intuitively understand why these conditions are relevant to the notion of implementability. Consider any implementable SCF. By the dominant strategy requirement of implementation, we must have strategy-proofness of $f_\succ$ for all $\succ$ in the support of the designer's belief. Hence the necessity of strategy-proofness is obvious. Moreover, if a message is a Nash equilibrium at a state $\theta$, it will continue to be a Nash equilibrium for states which share the profile of preferences in $\theta$, since the lower contour sets coincide exactly for each player across the states. So if two different alignments generate some common set of single-peaked preferences, whenever these preferences are held by agents, the outcomes must also be exactly the same. By the requirement of Nash implementation in the public information, we must then have shared-monotonicity.

It turns out that these two conditions, along with NVP, are also sufficient for any SCF to be implementable. We use a fairly standard mechanism to show this. Agents are asked to report the state i.e. an alignment and a preference single-peaked according to that alignment, and also their preferred alternative and an integer. If at least N-1 agents agree on the alignment, then that alignment is fixed as the public state. If not, the agent reporting the highest integer gets his preferred alternative.

\begin{theorem}\label{neccsuffthm}
If a social choice function $f$ is implementable in mixed information equilibria, then $f_{\succ}$ is strategy-proof for all $\succ$ in the support of $\mu$, and $f$ is shared-monotonic.

Conversely, if a social choice function $f$ is such that $f_{\succ}$ is strategy-proof for all $\succ$ in the support of $\mu$, $f$ is shared-monotonic, and $f$ satisfies no veto power, then it is implementable in mixed information equilibria.
\end{theorem}

\section{Implementability of unanimous and anonymous rules}


Now, after identifying the right solution concept and the right set of tools to check for it, we turn towards identifying what implementable rules actually look like. Recall that our motivation was to apply robustness requirements on Moulin (1980), looking at unanimous and anonymous rules. This classic paper led to the characterization that a rule is strategy-proof, unanimous and, anonymous if and only if it is a generalized median rule. These rules have an intuitive and appealing visual representation. First, place the alternatives on a line left to right according to the underlying alignment. Then, place $n-1$ ``phantom top" on the alternatives on the line - each generalized median rule is associated with a unique placement of these phantom tops. Given any profile of preferences of the agents, identify each agent's most preferred alternative, or ``top", and mark them on the line. The outcome of the generalized median rule is the median of these $n$ agent tops and the $n-1$ phantom tops.

Part of the appeal of these properties lies in their normative nicety. If all the agents like the same alternative, it would be quite reasonable for a social rule to assign them that alternative. Equality amongst the agents as well is a fairly standard requirement. The intuitive nature of the generalized median rules also adds to the appeal of these properties, since in essence they are tied to these rules by the characterization. However, importantly, the class of generalized median rules is vast. In fact, if we have $m$ alternatives and $n$ agents, his result gives us ${m+n-2}\choose{n-1}$ rules to choose from. Our results will be able to identify which of these will be robust to designer uncertainty in the underlying alignment.

At this point, we should be more specific by what we mean by designer uncertainty. We assume that the designer has some belief over the possible alignments according to which the domain of agents' preferences is single-peaked. The result in the previous subsection identifies general conditions on implementable SCFs. The strength of these conditions, of course, is dependent on the kind of belief that the designer has. Varying her belief will vary the type of SCFs that will be implementable. In order to be robust to changes in belief probabilities, we work with the support of the belief i.e. the designer only uses information about the alignments that she places a positive probability on, and not the specific probability that she places on them.

It will also be useful to state Moulin's result here, since we will draw on it later.

\begin{theorem}[Moulin (1980)]
In a single-peaked domain, a social choice function is strategy-proof, anonymous and unanimous if and only if it is a generalized median rule.
\end{theorem}

A generalized median rule specifies $n-1$ ``phantom" voters who each vote for a given alternative regardless of the preferences of the agents. The rule chooses the median of the $n$ most-preferred alternatives of the agents and these $n-1$ phantom votes. An example is shown in Figure \ref{fig:MoulMed} - note that $P_1$ and $P_2$ are the top preferences of the agents and $F_1$ is the fixed phantom vote that identifies the rule.

\begin{figure}[htp]
    \centering
    \includegraphics[width=8cm]{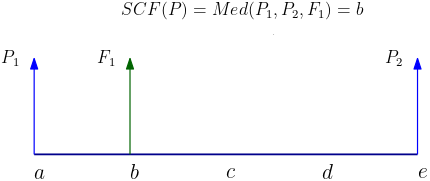}
    \caption{Moulin's generalized median SCF.}
    \label{fig:MoulMed}
\end{figure}

For us, the following definitions of unanimity and anonymity will be applicable.

\begin{definition}
Let $P(1)$ identify the most preferred alternative in $P$, also called its peak. The SCF $f$ is \textbf{unanimous} if for all preference profiles $(P_i)$ such that $P_i(1) = a$ for some $a \in X$ and for all $i \in N$, we have $f(\succ, P) = a$.
\end{definition}

\begin{definition}
The SCF $f$ is \textbf{anonymous} if for all preference profiles $(P_i)$ and for all bijective functions $\sigma:N \to N$, we have $f(\succ,(P_i)) = f(\succ,(P_{\sigma(i)}))$.
\end{definition}

Since implementable SCFs must be strategy-proof for each $\succ$, it is clear that adding unanimity and anonymity implies that implementable SCFs must have some median rule for each $f_{\succ}$. A slightly more subtle point relates to the property of no veto power. Consider what happens when there are no phantoms of a median rule at one of the two extreme ends. Then, one can always construct a profile where a single agent would indeed have a veto - for example with all but one agent having their top preference at the end without a phantom. At the same time, having at least one phantom on both ends ensures that there is no veto at any profile.

\begin{observation}
If an anonymous and unanimous SCF $f$ is implementable, then $f_{\succ}$ must be a median rule for all alignments $\succ$ in the support of $\mu$. It satisfies no veto power if and only if each $f_{\succ}$ places at least one phantom on the extreme ends of $\succ$. We call such SCFs \textbf{no veto projected median SCFs (NVPMS)}.
\end{observation}

\begin{figure}[htp]
    \centering
    \includegraphics[width=8cm]{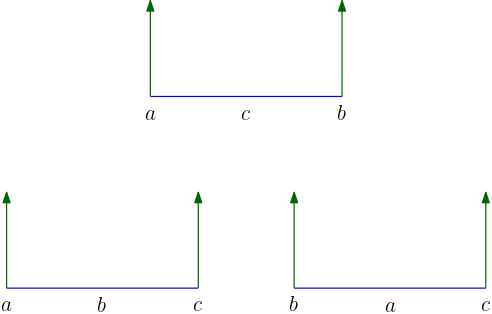}
    \caption{An NVPMS. Note that the three alignments that the designer believes possible are $acb$, $abc$, and $bac$. The green arrows represent position of the phantoms.}
    \label{fig:NVPMS}
\end{figure}

We present five results here, progressing in a natural way over possible beliefs of the designer. First, suppose the designer places positive probabilities over exactly two alignments, each of which is the reverse of the other. We may think of this as a good starting point, since the single-peaked domains produced by each alignment is exactly the same.

\begin{observation}
Let $\hat{\succ}^{-1}$ denote the alignment where the ordering of alternatives of $\hat{\succ}$ is exactly reversed. If $supp(\mu) = \{\hat{\succ}, \hat{\succ}^{-1} \}$ for some $\hat{\succ} \in \mathcal{P}$, then any NVPMS with $f_{\hat{\succ}} = f_{\hat{\succ}^{-1}}$ is implementable.
\end{observation}

When the designer is dealing with such a belief, any report of preference profiles could have been produced by either of the alignments. Shared-monotonicity then kicks in with full force - it must be the case that the SCFs chosen for each alignment must then be exactly the same.

This leads us to ask whether there might be cases where shared-monotonicity might have no bite at all. Indeed, whenever the two sets of single-peaked preferences for two alignments are disjoint, shared-monotonicity is satisfied vacuously, and we have the following result.

\begin{proposition}\label{delta}
Let $\mathcal{D}_{\succ}$ be the set of all preferences that are single-peaked according to $\succ$. If for all $\succ, \succ^{'} \in supp(\mu)$, $\mathcal{D}_{\succ} \cap \mathcal{D}_{\succ^{'}} = \emptyset$, then any NVPMS is implementable.
\end{proposition}

This is a very permissive result, since it gives a free hand to the designer to choose any median rule for each alignment, as long as it places at least one phantom on both ends of each alignment. 

We would quite naturally be interested to know the circumstances under which the designer would be much more constrained than this in her choices. Consider then the case where the designer is uncertain entirely about the underlying alignment, that is, she places a positive probability over every possible alignment in $\mathcal{P}$. This would be the case when the designer wants her mechanism to robust to any possible underlying alignment, for example when designing a voting system where the voters may end up having any common concern, or even any change in the alignment over time. 

The result for such a belief structure requires a few notations, which we present before the theorem.

\begin{definition}
Given $|A| = 3$, we call an SCF a \textbf{symmetric order-statistic SCF} if for any alignment $\succ \ \in \{abc, bca, cab\}$ where $a,b,c \in X$, $f{\succ}$ places $k$ phantoms on the leftmost alternative and $(n-1)-k$ on the rightmost alternative, with $1 \leq k \leq (n-2)$. Note that in alignment $abc$, $a$ is the leftmost alternative and $c$ is the rightmost alternative.
\end{definition}

\begin{figure}[htp]
    \centering
    \includegraphics[width=8cm]{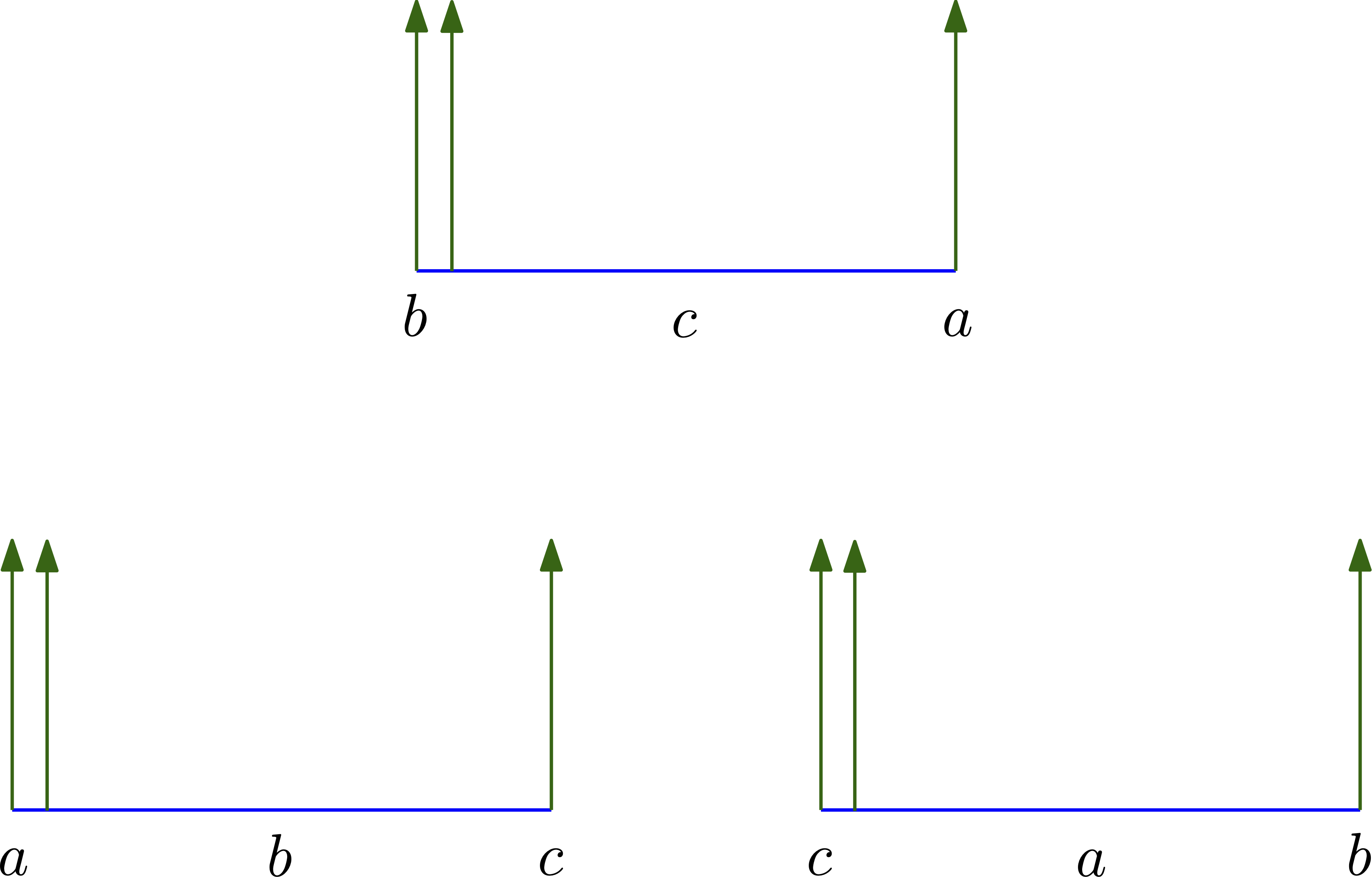}
    \caption{A symmetric order-statistic SCF. Note that for each alignment, the phantoms are at the ends and in the same pattern across alignments.}
    \label{fig:OrdStat}
\end{figure}

\begin{definition}
Given $n$ is odd, we call an SCF a \textbf{true median SCF} if $\forall \succ \ \in supp(\mu)$, $f_{\succ}$ chooses the median of the reported peaks of the agents.
\end{definition}

\begin{figure}[htp]
    \centering
    \includegraphics[width=8cm]{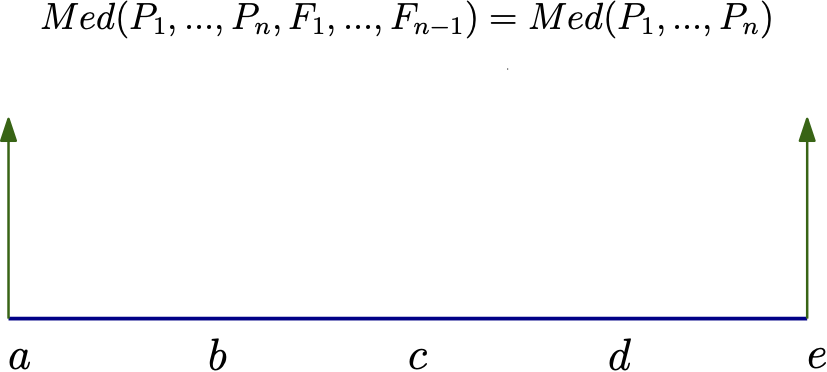}
    \caption{A true median SCF. Note that all the phantoms must be equally distributed at the two ends.}
    \label{fig:TMM}
\end{figure}


\begin{theorem}\label{FullSupp}
Suppose $supp(\mu) = \mathcal{P}$ and $|A| \geq 3$. Then
\begin{enumerate}
    \item If $|A| = 3$, then an SCF is implementable if and only if it is a symmetric order-statistic mechanism.
    \item If $n$ is even and $|A| > 3$, then there are no implementable SCFs.
    \item If $n$ is odd and $|A| > 3$, then an SCF is implementable if and only if it is the true median mechanism.
\end{enumerate}
\end{theorem}

The surprising thing about this result is how it picks a handful of median rules from the vast class identified by Moulin. From ${m+n-2}\choose{n-1}$ possible rules, these have been whittled down to just $1$ when we have $m>3$ and $n$ is odd. What this means is that if the designer actually has some (possibly very small) non-zero belief over all possible alignments, she should choose the true-median so that she can have this additional bulwark of implementability in mixed information equilibria. This result (and others in this study) can thus also be seen as identifying special mechanisms within the class of median rules, so that designers can be assisted in choosing from within it. 

Let us now try to address what happens when we do not have the full support, but there are still alignments with common preferences in their single-peaked domains. Because a general characterization is bound to be messy, we consider an important special case.

\begin{definition}
Let $T_{\succ, \succ^{'}} = $ $\{x \in X \ | \ x = P(1) $ for some $ P \in \mathcal{D}_{\succ} \cap \mathcal{D}_{\succ^{'}} \}$. We say that the belief $\mu$ has \textbf{constant shared peaks} if $T_{\succ, \succ^{'}} = T$ for all $\succ, \succ^{'} \in supp(\mu)$, where $T \subseteq X$ is a fixed subset of the set of alternatives.
\end{definition}

\begin{figure}[htp]
    \centering
    \includegraphics[width=8cm]{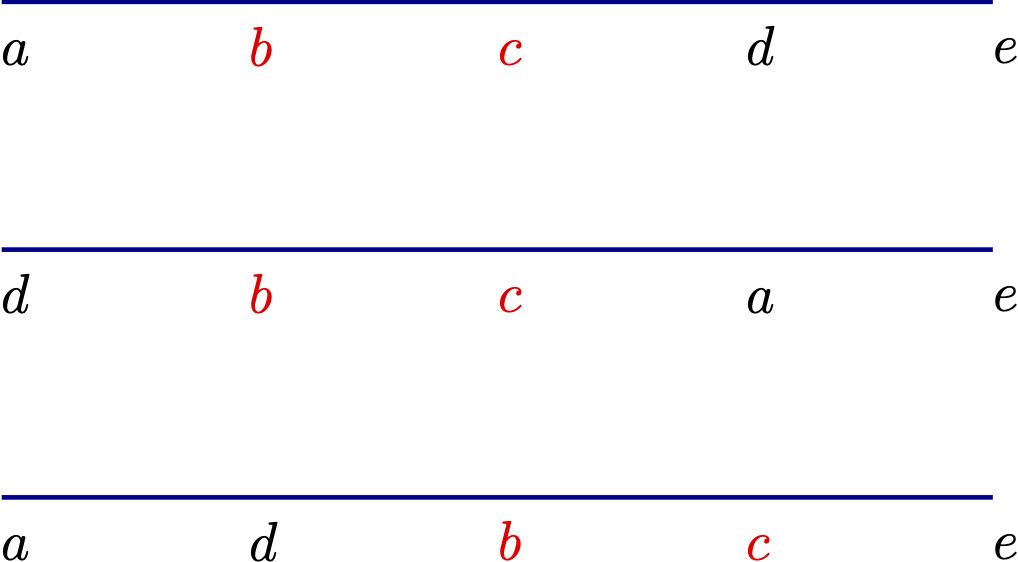}
    \caption{A set of alignments that have constant shared peaks. All preferences that are shared across these alignments have peaks on the red colored alternatives.}
    \label{fig:SP}
\end{figure}

In this case, as we show in a lemma that we use to prove our theorems, the shared peaks must be a contiguous subordering of each alignment, so one can imagine these tops as being lumped together in the same order, and then the other alternatives moving around to the left and right of this lump to generate the entire support. For such beliefs, there is a fairly large possibility result.

\begin{proposition}\label{SameShare}
Suppose $\mu$ has constant shared peaks. Then $T$ is a contiguous subordering of each alignment in $supp(\mu)$. Denote this subordering by $R$. An SCF is implementable if it is a projected median mechanism that, given any $\succ, \succ^{'} \ \in supp(\mu)$, places
\begin{enumerate}
    \item the same number of phantoms before and after $R$ for both $\succ$ and $\succ^{'}$
    \item at least one phantom on either end of both $\succ$ and $\succ^{'}$
    \item the same pattern of phantoms on $R$ for both $\succ$ and $\succ^{'}$
\end{enumerate}
\end{proposition}

Proposition \ref{SameShare} admits an interesting corollary. Suppose $|supp(\mu)| = 2$. In this case, regardless of which alignments are part of $supp(\mu)$, either the two alignments share a given set of preferences or they do not. Thus, the premise of the proposition is met, and we can identify implementable SCFs. $|supp(\mu)| = 2$ includes practically important situations such as when the designer is uncertain just about the relative alignment of two adjacent alternatives.

Finally, we show that if we impose a fairly weak additional assumption, we again get a very permissive result. Dutta and Sen (2012) \cite{dutta2012nash} introduced the concept of a ``partially honest" agent, who strictly prefers to reveal the true state as long as he is not worse off. 

\begin{definition}
An agent $i$ is partially honest if given that $\hat{\succ}$ is the true underlying alignment, and $m,m^{'}$ are such that $m_{i} = (\_,\hat{\succ},\_)$ and it is not the case that $g(m^{'}) P_i g(m)$, then he strictly prefers to send $m$.
\end{definition}

They show that having just one partially honest agent can make it possible to Nash implement SCFs that satisfy only No Veto Power, doing away with the otherwise necessary Maskin Monotonocity requirement. Unsurprisingly, this assumption is strong in our setting as well - in the presence of just one partially honest agent, we are able to implement a very large class of SCFs, regardless of the belief structure. 

\begin{proposition}\label{PartHon}
Let $|supp(\mu)| \geq 3$. Suppose that at least one agent is partially honest. Then, any NVPMS is implementable.
\end{proposition}

\section{Discussion}
The primary innovation of this paper is the study of a setting in which the state of the world has public and private components, with the designer being uninformed about both components and agents informed of the public state and their own private state. Settings like this merit practical interest, since not all information available to an agent can be neatly bucketed into either something only privately observed or only publicly observed. The solution concept that we use is relevant for our specific case, but there may be others or generalizations of the same which could lead to interesting results.

We focus on the single peaked domain since our interest is in answering questions related to robustness when the designer is not certain about the underlying alignment of the domain. Our general result on implementation in mixed information equilibria, which involves Nash implementation in the public component and dominant strategy implementation in the private component, identifies necessary and sufficient conditions on SCFs for them to be implementable. This is in and of itself useful in practice, since these conditions can be directly checked and it can be verified whether the SCF can be implemented. 

It is for anonymous and unanimous rules that we go further and identify implementable SCFs. Fixing the public part of the state i.e. the underlying alignment, implementability necessitates that the SCF behave as a Moulin generalized median mechanism. Implementability could also impose further constraints on which median mechanisms these can be, depending on which alignments are possible. For some sets of alignments, there are no additional constraints imposed, while if the set of alignments is the set of all possible alignments, then only one median mechanism is admissible for each alignment. This is how we answer the robustness questions - depending on which alignments the designer thinks are possible, she can identify which SCF she should choose from all of Moulin's median mechanisms.

A more general question to tackle in this setting would be to identify the structure of unanimous implementable rules. As alluded to earlier, it is known that unanimous strategy-proof rules must be ``min-max" rules, and the same questions of robustness to designer uncertainty can be asked about these rules. Here our choice was to investigate the generalized median rules since they have an intuitive visualization and also because anonymity is a fairly standard normative requirement. The literature post \cite{Moulin1980} has also focused a great deal on median rules for similar reasons.

\newpage

\bibliography{SingPeak}

\begin{appendices}

\section{Proofs}

\begin{proof}[\textbf{Proof of Theorem \ref{neccsuffthm}}]
Let $f$ be an implementable SCF, and let $(M,g)$ be the mechanism that implements it. 

Let $\succ$ be an alignment in the support of $\mu$. Fix an agent $i$. Let $P_i$, $P_{i}^{'}$ and $P_{-i}^{'}$ be arbitrary preferences single peaked according to $\succ$. Then, there exists an $m^*$ such that $m^{*}_{i} = (l^{*}_i, \succ, P_{i})$, $m^{*}_{-i} = (l^{*}_{-i}, \succ, P^{'}_{-i})$, and $g(m*) = f(\succ, P_i, P_{-i}^{'})$. Also, there exists an $m^{**}$ such that $m^{**}_{i} = (l^{**}_i, \succ, P_{i}^{'})$, $m^{**}_{-i} = (l^{**}_{-i}, \succ, P^{'}_{-i})$, and $g(m^{**}) = f(\succ, P_{i}^{'}, P_{-i}^{'})$. By dominant strategy implementation in private information, $g(m^{*}_{i}, (l^{*}_{-i}, \succ, P_{-i}^{'})) \ R_i \ g((l^{**}_{i}, \succ, P_{i}^{'}), (l^{**}_{-i}, \succ, P_{-i}^{'}))$, which implies $g(m^*) R_i g(m^{**})$, and hence we have $f(\succ, P_i, P_{-i}^{'}) R_i f(\succ, P_{i}^{'}, P_{-i}^{'})$. Since $P_{i}^{'}$ and $P_{-i}^{'}$ were arbitrary, this means that $f_{\succ}$ is strategy-proof.

Let $(\succ, P)$ and $(\succ^{'}, P)$ be two states that share preferences. Let $m_{i}^* = (l^{*}_i, \succ, P_{i})$ be the implementing message in state $(\succ, P)$ and $m^{**}$ in $(\succ^{'}, P)$. We must have that deviating with respect to the underlying alignment leads to some outcome in the lower contour set for each agent $i$. Then, at state $(\succ^{'}, P)$ as well deviations from $m^*$ will lead to some outcome in the lower contour set for each agent $i$, and hence $m^*$ will be a Nash equilibrium at $(\succ^{'}, P)$. By Nash implementation in public information, $f(\succ, P) = g(m^*) = g(m^{**}) = f(\succ^{'}, P)$.

For the other direction, let $f$ be such that $f_{\succ}$ is strategy-proof for all $\succ$ in the support of $\mu$, $f$ is shared-monotonic, and $f$ satisfies no veto power. Consider the following mechanism:
\begin{itemize}
    \item $M_i = N \times X \times \mathcal{K}$ for all $i$, where $N$ is the set of natural numbers.
    \item if messages are of the form $m_i = (l_i, x_i, \succ, P_i)$ for all except at most one agent i.e. at least $n-1$ agents agree on the alignment, $g(m) = f(\succ, P)$.
    \item if not, select the agent with the lowest index amongst those sending the highest natural number, say $j$. $g(m) = x_j$.
\end{itemize}
To verify dominant strategy implementation in private information, note that given $\succ$, truth-telling is a dominant strategy since $f_{\succ}$ is strategy proof. Unilateral deviations along $N \times X$ do not change the outcome. To verify Nash implementation in public information, note first that unilateral deviations for the alignment do not change the outcome. Second, if all but one agent $i$ report the same alignment but $i$ reports a different alignment, we can have a Nash equilibrium in the public information only if the outcome chosen is the best alternative for all agents but $i$. But then, by NVP, this outcome coincides with $f(\succ, P)$. Finally, if all agents report a different alignment, by shared-monotonicity the outcome is the same as $f(\succ, P)$. 
\end{proof}



\begin{lemma}\label{ShareLem}[Consistency]
Suppose $\succ$ and $\succ^{'}$ are such that $\mathcal{D}_{\succ} \cap \mathcal{D}_{\succ^{'}} \neq \emptyset$. Let $T_{\succ, \succ^{'}} = \{x \in X \ | \ x = P(1) $ for some $ P \in \mathcal{D}_{\succ} \cap \mathcal{D}_{\succ^{'}} \}$. Then the following hold
\begin{enumerate}
    \item $|T_{\succ, \succ^{'}}| \geq 2$.
    \item Let $T_{\succ, \succ^{'}} = \{x_1, x_2,...,x_{|T|}\}$. Then $x_1, x_2,...,x_{|T_{\succ, \succ^{'}}|}$ are adjacent in both $\succ$ and $\succ^{'}$, and they are in the same order in both $\succ$ and $\succ^{'}$ \footnote{We reiterate that an alignment and its exact reverse are equivalent for our purposes, so this statement and the rest in this paper should be considered up to these equivalence classes.}.
\end{enumerate}
\end{lemma}

\begin{proof}[\textbf{Proof of Lemma \ref{ShareLem}}]
\begin{enumerate}
    \item Note that $T_{\succ, \succ^{'}} \neq \emptyset$. Let $x \in T_{\succ, \succ^{'}}$. Then, $\exists \ P \in \mathcal{D}_{\succ} \cap \mathcal{D}_{\succ^{'}}$ such that $x = P(1)$. Let $y = P(2)$. Consider a preference $P^{'}$ such that $P^{'}(1) = y$ and $P^{'}(2) = x$ and $P^{'}(i) = P(i) \ \forall \ i > 2$. Since this change continues to make the new ordering single-peaked with respect to whatever $P$ was single-peaked according to, we must have $P^{'} \in \mathcal{D}_{\succ} \cap \mathcal{D}_{\succ^{'}}$. Since this implies $y \in T_{\succ, \succ^{'}}$, we are done.
    \item Let $x, y \in T_{\succ, \succ^{'}}$ such that in $\succ$, $t$ lies between $x$ and $y$ for all $t \in T_{\succ, \succ^{'}}$. Now, let $z \in X$ be in between $x$ and $y$ in $\succ$. Let us check where $z$ will be in $\succ^{'}$. Suppose it is not in between the two. Without loss of generality, let it be to the left of both $x$ and $y$ in $\succ^{'}$. Since $y \in T_{\succ, \succ^{'}}$, $\exists \ P \in \mathcal{D}_{\succ} \cap \mathcal{D}_{\succ^{'}}$ such that $y = P(1)$. By single peakedness, $zPx$ in $\succ$ but $xPz$ in $\succ^{'}$, which is a contradiction. Thus, $z$ must be in between $x$ and $y$ in $\succ^{'}$ as well. \\
    It is straightforward to show by a similar argument that if both $z_1$ and $z_2$ are in between $x$ and $y$ in $\succ$, then the four must be in the same order in $\succ^{'}$ as well. Thus, the order between $x$ and $y$ is preserved across $\succ$ and $\succ^{'}$. \\
    Now, let $z$ be next to $x$ in $\succ$ and between $x$ and $y$. Since $x \in T_{\succ, \succ^{'}}$, $\exists \ P^{'} \in \mathcal{D}_{\succ} \cap \mathcal{D}_{\succ^{'}}$ such that $x = P^{'}(1)$. Consider a preference $P^{''}$ such that $P^{''}(1) = z$ and $P^{''}(2) = x$ and $P^{''}(i) = P^{'}(i) \ \forall \ i > 2$. We must have $P^{''} \in \mathcal{D}_{\succ} \cap \mathcal{D}_{\succ^{'}}$. In this manner, $z \in T_{\succ, \succ^{'}}$ for all $z$ in between $x$ and $y$, which proves our hypothesis. 
\end{enumerate}
\end{proof}

\begin{lemma}\label{2Lemm}[Symmetry]
Suppose $\succ$ and $\succ^{'}$ are such that $\mathcal{D}_{\succ} \cap \mathcal{D}_{\succ^{'}} \neq \emptyset$. Let $x,y \in T_{\succ, \succ^{'}}$ be such that $x$ and $y$ are adjacent to each other, with $x$ before $y$. Then, $f_{\succ}$ and $f_{\succ^{'}}$ must have the same number of phantoms with tops on or before $x$.
\end{lemma}

\begin{proof}[\textbf{Proof of Lemma \ref{2Lemm}}]
Let $P_x \in \mathcal{D}_{\succ} \cap \mathcal{D}_{\succ^{'}}$ denote the preference with $x$ at top and $P_y \in \mathcal{D}_{\succ} \cap \mathcal{D}_{\succ^{'}}$ denote the preference with $y$ at top. For any $k,l \in \{0,1,...,n\}$ such that $k+l=n$, we must have that reports of $k$ number of $P_x$ and $l$ number of $P_y$ preferences must lead to the same outcome across $\succ$ and $\succ^{'}$. If the number of phantoms before $x$ were not the same across the two alignments, this would not be possible.
\end{proof}

\begin{proof}[\textbf{Proof of Theorem \ref{FullSupp}}]
Since Lemma \ref{2Lemm} establishes that Symmetry is necessary, we first show that everything but the order statistic mechanisms with odd agents will violate Symmetry and, for $|A| > 3$, we show that all except the true median will violate it. Then we show that for these surviving mechanisms, NVP is necessary. Showing that the surviving mechanisms satisfy Symmetry and NVP completes the proof.

Our first claim is that phantoms cannot be placed on the interior for any alignment, where interior of an alignment is anything not at the two ends. Suppose the median SCF chosen for some alignment, say $\succ$, contains phantoms in the interior. Select one alternative in the interior which has a phantom on it, and call it $y$. Let the alternative to the left of it be $x$ and the one to its right be $z$. Suppose there are $k$ phantoms on or before $x$, and $l$ phantoms on or before $y$. By assumption, $k < l$. Consider the alignment, say $\succ^{'}$, with $z$ pulled to the left of $x$, but everything else in the same position. Since we have full support, this alignment is in the support of $\mu$. Clearly, $\succ$ and $\succ^{'}$ share preferences with tops at $x$ and $y$. Thus, there must continue to be $k$ phantoms before $x$ in $\succ^{'}$. Consider the alignment, say $\succ^{''}$, with $y$ and $z$ interchanged in $\succ$, with all other alternatives at the same position. $\succ$ and $\succ^{''}$ share preferences with tops at $y$ and $z$, so after $y$ in $\succ^{''}$, we must have $l$ phantoms. Now, note that $\succ^{'}$ and $\succ^{''}$ share preferences with tops at $x$ and $z$. Before $x$ in $\succ^{'}$, we have $k$ phantoms but after $x$ in $\succ^{''}$ we have at least $l$ (since $y$ is to the right of $x$), which is a contradiction.

Second, we claim that the number of phantoms must be the same on both ends for all alignments when $|A| > 3$. Suppose the median SCF chosen for some alignment, say $\succ$, contains $k$ phantoms at the left end. Let $a$, $b$, $c$, $d$ be the first four alternatives in $\succ$, in that order. Consider the alignment, say $\succ^{'}$, with $a$, $d$, $b$, $c$ as the first four alternatives, in that order, and the rest the same as $\succ$. Clearly, $\succ$ and $\succ^{'}$ share preferences with tops at $b$ and $c$, so $\succ^{'}$ must have $k$ phantoms on $a$.  Consider the alignment, say $\succ^{''}$, with $c$, $b$, $a$, $d$ as the first four alternatives, in that order, and the rest the same as $\succ$. Clearly, $\succ^{'}$ and $\succ^{''}$ share preferences with tops at $a$ and $d$, so $\succ^{''}$ must have $k$ phantoms on $c$. Now, note that $\succ^{'}$ and $\succ^{''}$ share preferences with tops at $a$ and $b$, so in $\succ$, there must be $k$ phantoms on the right end. Thus, we have ruled out everything but the true-median. 

Now, we note that No Veto Power (NVP) is also a necessary condition for the surviving mechanisms under full support. In the class of order statistics mechanisms, a mechanism does not satisfy NVP if for some $\succ$, the SCF chosen for $\succ$ has all phantoms placed on one or the other side. Suppose such a $\succ$ exists and without loss of generality, let all phantoms be on the left end. Let $a,b,c$ be the first three alternatives in $\succ$. Consider a profile of preferences with $n-1$ peaks at $a$ and $b$, and $1$ peak at $c$ (call this agent $i$). Consider a $\succ^{'}$ where $c,a,b$ are the first three alternatives and everything else is the same as $\succ$. By Symmetry, $\succ^{'}$ must also have all the peaks at $c$. Then, by misreporting to a preference $\succ^{'}$ with top at $c$, agent $i$ can profitably deviate.

We close by noting that the true median and the symmetric order statistic mechanisms satisfy both NVP and Symmetry. Consider the true median mechanism. For any pair of alternatives, Symmetry is satisfied by definition since the number of phantoms is the same on both ends of every alignment. For any profile of preferences with $n-1$ them on the same alternative, this alternative must be chosen since there will be at least one phantom on each end. Consider the symmetric order statistic mechanism. For any pair $a,b$, they are shared across exactly two alignments, and in the same order. Thus, symmetry is guaranteed. NVP follows by the fact that they have at least one phantom on either side.
\end{proof}

\begin{proof}[\textbf{Proof of Proposition \ref{SameShare}}]
This follows directly from Lemmas \ref{ShareLem} and \ref{2Lemm}. Note that Consistency implies that the shared tops must be a contiguous subordering of each alignment in $S$, and Symmetry then implies parts $1.$ and $3.$ of the proposition. $2.$ is a consequence of NVP, as noted earlier.
\end{proof}

\begin{proof}[\textbf{Proof of Proposition \ref{PartHon}}]
Consider the following indirect mechanism $\mathcal{M}$: each agent sends a message $m_i = (\succ_i, P_i, f_i, z)$ such that $\succ_i \in \mathcal{P}$, $P_i \in \mathcal{D}_{\succ_i}$, $f_i$ is a median SCF, and $z$ is an integer. Let $x_i$ be the top of $P_i$.
\begin{enumerate}
    \item If at least $N-1$ agents send $\succ_i = \succ^{*}$, run any median mechanism $\mathcal{M}^c$ satisfying No Veto Power (NVP) using $\succ^{*}$ as the underlying alignment and $x_i$ as the peaks, and
    \item otherwise, choose the agent with the lowest index agent, say $j$, amongst those sending the largest integer in their message, and run the median mechanism $\mathcal{M}^j$ dictated by $f_j$ and $\succ_j$, using $x_i$ as the peaks.
\end{enumerate} 

Suppose all agents send $\hat{\succ}$ as part of their message. Then, $\mathcal{M}^c$ is run using the correct alignment and all the equilibria are the same as a strategy-proof, unanimous, and anonymous direct mechanism. Deviations in reporting of $\succ_{i}$ do not lead to a change in the outcome, and hence they cannot be profitable.

Suppose all agents send $\succ^{'} \neq \hat{\succ}$ as part of their message. Since there exists an agent with preference for honesty, this agent will deviate to sending $\hat{\succ}$ since that doesn't change the outcome, and hence this cannot be an equilibrium message profile.

Suppose there are two alignments $\succ^{'}$ and $\succ^{''}$ sent by the agents. Consider two cases:

\begin{adjustwidth}{1 cm}{}
\textit{Case 1}: The true peaks of all but one agents are at the same alternative, say $a$. If at least one of the agents with the true peak at $a$ sends an alignment $\succ^{'}$ while the others send $\succ^{''} \neq \succ^{'}$, then the agent with the true peak not at $a$ can deviate to some $\succ^{'''} \notin \{\succ^{'}, \succ^{''}\}$ and get his true peak selected. 
\\
Suppose instead that all the agents with the true peak at $a$ send the same $\succ^{'}$ and the remaining agent sends $\succ^{''} \neq \succ^{'}$. Note that for such a message to be a Nash equilibrium, it must be that $a$ is the outcome, since if some $b \neq a$ is chosen, then, one of the agents with peak at $a$ can deviate to report some $\succ^{'''} \notin \{\succ^{'}, \succ^{''}\}$ and get $a$ chosen, which is a profitable deviation. By NVP, the alternative chosen in the dominant strategy equilibrium of $\mathcal{M}^c$ under knowledge of $\succ$ is $a$ as well.
\\
\textit{Case 2}: If the true peaks of fewer than $N-1$ agents are the same, then there will be at least two agents who do not get their preferred alternative. At least one of these agents can profitably deviate by sending a larger integer and a different alignment $\succ^{'''}$ (along with the appropriate $f_j$) to get his preferred alternative given the messages of others.\end{adjustwidth}
\end{proof}

\end{appendices}

\end{spacing}
\end{document}